\documentclass[11pt,onehalfspacing]{article}
\usepackage{graphicx}
\usepackage[dvipsnames,svgnames,x11names,hyperref]{xcolor}
\usepackage{amsthm, amsmath, amssymb}
\usepackage{thmtools, thm-restate}
\usepackage{xfrac}
\usepackage{dsfont}
\usepackage[T1]{fontenc}
\usepackage{enumitem}
\usepackage{fullpage}
\usepackage{xspace}%
\allowdisplaybreaks

\usepackage{thmtools, thm-restate}
\usepackage[colorlinks,pagebackref]{hyperref}
\hypersetup{
  linkcolor=[rgb]{0,0,0.4},
  citecolor=[rgb]{0, 0.4, 0},
  urlcolor=[rgb]{0.6, 0, 0}
}
\usepackage{mathrsfs,bbm}
\usepackage{todonotes}
\usepackage{tikz}
\usetikzlibrary{decorations.pathreplacing,backgrounds}
\usepackage{multirow}
\usepackage{setspace}

\def\Bezout{B\'{e}zout}

\newcommand{\N}{\mathbb N}

\newcommand{\F}{\mathbb F}

\newcommand{\cI}{{\mathcal{I}}}   
   \newcommand{\cL}{{\mathcal{L}}}
   
\newcommand{\cO}{{\mathcal{O}}}   \newcommand{\cP}{{\mathcal{P}}}

\renewcommand{\epsilon}{\varepsilon}
\newcommand{\eps}{\epsilon}

\newif\ifdraft
\drafttrue
\ifdraft
\newcommand{\ASnote}[1]{\textcolor{Blue}{\guillemotleft AS: #1 \guillemotright}}

\else
\newcommand{\ASnote}[1]{}
\fi

\newcommand{\dlog}{\mathsf{dlog}}

\newcommand{\calI}{\mathcal{I}}

\newcommand{\calL}{\mathcal{L}}
\newcommand{\calP}{\mathcal{P}}

\hypersetup{
  linkcolor=[rgb]{0,0,0.4},
  citecolor=[rgb]{0, 0.4, 0},
  urlcolor=[rgb]{0.6, 0, 0}
}

\newcommand{\ignore}[1]{}
\DeclareMathOperator{\Char}{char}

\usepackage[nameinlink,capitalise]{cleveref}

\usepackage{amsthm}
\usepackage{thmtools,thm-restate}

\numberwithin{equation}{section}
\declaretheoremstyle[bodyfont=\it,qed=\qedsymbol]{noproofstyle}

\declaretheorem[name=Observation,numbered=no]{observation*}

\declaretheorem[numberlike=equation]{theorem}

\declaretheorem[name=Theorem,numbered=no]{theorem*}

\declaretheorem[numberlike=equation]{lemma}
\declaretheorem[name=Lemma,numbered=no]{lemma*}

\declaretheorem[numberlike=equation]{corollary}
\declaretheorem[name=Corollary,numbered=no]{corollary*}

\declaretheorem[numberlike=equation]{proposition}
\declaretheorem[name=Proposition,numbered=no]{proposition*}

\declaretheorem[name=Claim,numbered=no]{claim*}

\declaretheorem[numberlike=equation]{conjecture}
\declaretheorem[name=Conjecture,numbered=no]{conjecture*}

\declaretheorem[name=Question,numbered=no]{question*}

\declaretheoremstyle[bodyfont=\it,qed=$\lrcorner$]{defstyle} 

\declaretheorem[numberlike=equation,style=defstyle]{definition}
\declaretheorem[unnumbered,name=Definition,style=defstyle]{definition*}

\declaretheorem[unnumbered,name=Example,style=defstyle]{example*}

\declaretheorem[unnumbered,name=Notation=defstyle]{notation*}

\declaretheorem[unnumbered,name=Construction,style=defstyle]{construction*}

\declaretheorem[numberlike=equation,style=defstyle]{remark}
\declaretheorem[unnumbered,name=Remark,style=defstyle]{remark*}

\title{Bounds on polynomial approximations of the discrete logarithm function}
\author{Swastik Kopparty\thanks{
Department of Mathematics and Department of Computer Science, University of Toronto, Canada. Email: \texttt{swastik.kopparty@utoronto.ca}.   Research partially supported
by an NSERC Discovery Grant.}
\and Rishabh Kothary\thanks{Department of Computer Science, University of Toronto, Canada. Email: \texttt{rishabh.kothary@mail.utoronto.ca}. Research partially supported by an NSERC Discovery Grant} \and Ashutosh Shankar\thanks{Tata Institute of Fundamental Research, Mumbai, India. Email:\texttt{ashutosh.shankar@tifr.res.in} Research conducted while the author was visiting the University of Toronto.}}

\begin{document}

\maketitle

\begin{abstract}
We revisit a classical result of Coppersmith and Shparlinski~\cite{CS00} showing bounds on how well polynomials can approximate the discrete log function. They showed that for any univariate polynomial $f(X) \in \F_p[X]$ with $\deg(f) \leq D$, we have:
$$ |\{ x \in \F_p^\times \mid f(x) = \dlog(x) \}|  \leq \sqrt{2 D p}.$$
We show two improvements:
\begin{itemize}
\item For all $D$, we improve the upper bound above to $(1+o(1)) \cdot \sqrt{Dp}$. In particular, we give nontrivial upper bounds on the approximation even for $D > p/2$.
This new upper bound is the {\em Johnson bound} from the theory of list-decoding of error-correcting codes, and in fact our proof uses a connection to that problem, as well as other
combinatorial ideas.

\item For $D = 1$, where both the above bounds are $\cO(\sqrt{p})$, we give an asymptotically improved upper bound of $\cO(p^{2/5})$. This is based on the sum-product phenomenon of additive combinatorics.
\end{itemize}

We also improve the~\cite{CS00} bounds on the number of  points of the form $(x, \dlog(x)) \in \F_p^2$ 
that any bivariate polynomial $Q(X,Y) \in \F_p[X,Y]$ with $\deg(Q) \leq D$ can vanish on,
with our new bounds approaching the analogue of the Johnson bound for this problem.
\end{abstract}

\section{Introduction}

The discrete logarithm function has been a central object in cryptography since
the work of Diffie and Hellman~\cite{DiffieHellman}. It plays a prominent role in 
innumerable cryptographic constructions. It is therefore of great importance to understand its algebraic and combinatorial characteristics, with the goal of identifying any computational weakness.

This paper is concerned with the best {\em low-degree approximation} to the discrete log function in the finite field $\F_p$, with prime $p$. Concretely, for a given degree bound $d$, we are interested in determining an upper bound $A$ such that for every polynomial $f(X) \in \F_p[X]$ of degree at most $D$,
we have:
$$ | \{ x \in \F_p \mid   f(x) = \dlog(x) \}| \leq A.$$
(Thus we are talking about approximation in terms of the number of points with exact equality, not any numerical approximation notion.)

This question was first studied by Coppersmith and Shparlinski~\cite{CS00}, who showed that
for any $f(X) \in \F_p[X]$ of degree at most $D$, 
$$ | \{ x \in \F_p \mid   f(x) = \dlog(x) \}| \leq \sqrt{2 D p}.$$

Subsequent work of Winterhof~\cite{Winterhof} extended this to non-prime finite fields $\F_q$.
Even further extensions, to fractional linear maps and beyond were given by
Kiltz and Winterhof~\cite{KW}, Isik and Winterhof~\cite{IW}, and Bienvenu and Winterhof~\cite{BW}. A different extension, of Niederreiter and Winterhof~\cite{NW}, gave bounds on the largest interval of agreement.
However there has been no improvement of the Coppersmith-Shparlinski bound over prime fields for any $D$.

Our first main result gives an improved bound. We show that 
for any $f(X) \in \F_p[X]$ of degree at most $D$, 
$$ | \{ x \in \F_p \mid   f(x) = \dlog(x) \}| \leq (1+ o(1)) \cdot \sqrt{D p}.$$
In particular, this gives nontrivial bounds even in the setting of $D > \frac{p}{2}$.

The quantity $\sqrt{Dp}$ in the upper bound is of special significance; it is the {\em Johnson agreement} parameter in the list-decoding of Reed-Solomon codes. 
The Johnson bound for Reed-Solomon codes states that if $g : \F_p \to \F_p$ is a function,
and there exist distinct polynomials $f_1(X), \ldots, f_L(X)$ of degree at most $D$ 
such that $|\{ x \mid f_i(x) = g(x) \}| \geq (1+ \epsilon) \sqrt{Dp}$, then $L \leq \cO_{\epsilon}(\sqrt{p/D})$.
There are two steps to our proof: (1) we show that if some polynomial $f(X)$ approximates $\dlog$ well, then there are many polynomials $f_1(X), f_2(X), \ldots, f_L(X)$ of the same degree
that also approximate $\dlog$ quite well, and
(2) using the list-decodability of Reed-Solomon codes to get a contradiction.

The first step uses the almost-homomorphism property of $\dlog$: $\dlog(x) + \dlog(y) \in \dlog(xy) + \{0,-1\}$, where the last term is because of the mismatch between the sizes of $\F_p$ and $\F_p^\times$.  The extra complexity introduced by the ``almost'' aspect is handled by some shifting arguments, along with Szemeredi's theorem from additive combinatorics.

One consequence of our approach is that improved bounds on the list-decoding of Reed-Solomon codes, with sublinear list size $L$, will give improved bounds for our problem. We think this is strong evidence in favour of the existence of better upper bounds\footnote{We remark that the very recent breakthrough of Brakensiek, Chen, Putterman, Zhang and Zheng \cite{brakensiek2026algorithmic} on list-decoding Reed-Solomon codes over $\F_p$ (see also the paper of Harsha, Kumar and Saptharishi~\cite{harsha2026reed}) does not yield anything yet because of the large polynomial list sizes in their result.}.
In fact, we believe that the true upper bound is much smaller:
\begin{conjecture}
    For any $f(X) \in \F_p[X]$ of degree at most $D$,
    $$ | \{ x \in \F_p \mid   f(x) = \dlog(x) \}| \leq  (1 + o(1) ) \cdot  D \cdot \frac{\log p}{\log (D+1)},$$
    as $p \to \infty$.
\end{conjecture}
The quantitative form is motivated by the best known impossibility results for list-decodability of Reed-Solomon codes over prime fields.

Our second main result makes some further progress towards this conjecture for the case of $D = 1$, i.e, linear polynomials. We show that for any $f(X) \in \F_p[X]$ of degree at most $1$:
   $$ | \{ x \in \F_p \mid   f(x) = \dlog(x) \}| \leq \cO(p^{2/5}).$$
This improves on the $\Theta(\sqrt{p})$ bound of Coppersmith-Shparlinski by a polynomial factor.

Our proof of the second result uses the sum-product phenomenon in finite fields~\cite{BourgainKT2004}, and specifically its consequence for incidences
between lines and points in $\F_p^2$~\cite{BourgainKT2004,lewko2026szemer,miao2026szemer,yu2026two}.  The main technical construction comes from a 2-source extractor construction of Bourgain~\cite{Bourgain} (see also Rao's exposition~\cite{Rao}), who showed how to construct good two-source extractors from the discrete log function using the line-point incidence theorem in finite fields.

Finally, we also mention that Coppersmith and Shparlinski proved another bound on approximating
the discrete log function with ``algebraic functions''. Specifically, they showed that for any bivariate polynomial $Q(X,Y) \in \F_p[X,Y]$ of degree at most $D$:
$$ | \{ (x,\dlog(x)) \in \F_p^2 \mid  Q(x,y) = 0 \} | \leq \sqrt{3} \cdot D \sqrt p.$$
We again improve this bound to the analogous Johnson bound\footnote{Here the relevant Johnson bound is the following statement: if $S \subseteq \F_p^2$ is a set of size $N$, and $Q_1(X,Y), \ldots, Q_L(X,Y)$ are distinct irreducible polynomials of degree at most $D$, such that each $Q_i$ vanishes at at least $(1+\epsilon) \cdot D \sqrt{N}$ points of $S$, then $L \leq O_{\epsilon}(\sqrt{N}/D)$. This follows from \Bezout's theorem and the classical Johnson bound on set systems with restricted intersections. } threshold, namely $(1 + o(1)) \cdot D \sqrt p$.

One final related work in the same spirit that we would like to mention is the paper
of Kedlaya and Kopparty~\cite{KedlayaKopparty}. They gave limits on the approximation by polynomials of the (multi-valued) ``square root'' function in finite fields.

\paragraph{AI use:} The results were obtained without any AI involvement. The writing was done without any AI involvement. We used AI for literature search and for feedback on our writeup.

\paragraph{Organization of this paper:}
In the next section we set up some notation and preliminaries about list decoding and additive combinatorics. In \cref{sec:conjugates} we prove our main structural lemmas about conjugates, which yields many approximations from one approximation. Our results for univariate polynomial approximation appear in \cref{sec:univariate}. Our results for bivariate polynomial approximation appear in \cref{sec:bivariate}. Our results for degree $D = 1$ appear in \cref{sec:linear}.

\section{Notation and preliminaries} \label{sec:prelims}

We work in a finite field $\F_p$ for some fixed prime $p$. We will often use $N$ to denote $p-1$ for ease of notation. Fix a generator $g$ of the field. The discrete logarithm is defined with respect to this generator.

\begin{definition}[Discrete logarithm]
    For $x \in \F_p^*$, $\dlog(x)$ is defined to be the 
    unique $j \in \{0,1, \ldots, p-2\}$ with $g^j = x$ in $\F_p$.
\end{definition}

We will invoke the following well-known result from algebra.

\begin{theorem}[\Bezout's theorem \cite{fulton1989algebraic}]
\label{thm:Bezout}
    Let $f,g \in \F[X,Y]$ such that $\deg(f) \leq d$, $\deg(g) \leq e$ and $f,g$ are co-prime. Then, the number of common zeroes of $f$ and $g$ is at most $d\cdot e.$
\end{theorem}

We will also talk about sumsets. For a subset $V$ of an abelian group $G$,
we let $m \cdot V$ denote the $m$-fold sumset $V+V+ \ldots + V$.

\subsection{Johnson bound}



The Johnson bound is generally stated as a bound on the number of codewords within a certain radius of a received word, given that the code has some minimum distance. Simple forms of the Johnson bound are proved by counting methods. We use the following rearrangement of standard Johnson bound statements, for instance Theorem 4.2 of \cite{GoldreichRS00}, so it is phrased as a bound on the size of the agreement sets (i.e. a lower bound on the radius) given a list of codewords and the minimum distance of the code.

\begin{lemma}[Johnson Bound]\label{lem:Johnson-Bound}
    Let $T_1,...,T_L \subset [n]$ be such that $|T_i| \geq A$ for all $i \in [L],$ and for any $i \neq j \in [L]$, $|T_i \cap T_j| \leq r$, then
    \begin{equation*}
        A \leq \frac{n}{2L} + \sqrt{\frac{n^2}{4L^2} + nr}.
    \end{equation*}
\end{lemma}

\subsection{Theorems from additive combinatorics}

In the sections on low-degree univariate and bivariate polynomials (\cref{sec:univariate,sec:bivariate}), we will use results on arithmetic-progression-free sets. Szem\'eredi's Theorem states that a large enough set of natural numbers must contain a long arithmetic progression. Below we state a quantitative version due to Gowers.

\begin{theorem}[\cite{Gow01}]\label{thm:Szemeredi-Theorem-Gowers}
     For every $k \in \N$, there exists a constant $n_k$ such that for any $n \geq n_k$ and $S \subset [n]$ where $|S| \geq n\cdot (\log\log n)^{-2^{-2^{k+9}}}$, then $S$ contains a $k$ term arithmetic progression.
\end{theorem}



In the section on linear functions (\cref{sec:linear}), we will invoke the following bound on line-point incidences. Let $\calP \subseteq \F_p^2$ be a set of points and $\calL$ a set of lines in $\F_p^2$. An incidence is a pair $(a, \ell)$ of $a \in \calP$ and $\ell \in \calL$ with $a \in \ell$. The set of incidences is denoted $\calI(\calP, \calL)$. The first nontrivial upper bounds on the number of incidences over prime fields were obtained by Bourgain, Katz and Tao \cite{BourgainKT2004}. The main challenge that \cite{BourgainKT2004} had to overcome is that finite fields which contain subfields actually have point-line configurations with many incidences; a nontrivial upper bound needs to use the absence of subfields of $\F_p$. The very recent breakthrough works of Lewko \cite{lewko2026szemer} and Miao and Xie \cite{miao2026szemer} gave optimal quantitative estimates for the number of incidences (see also the elegant paper of Yu \cite{yu2026two}). We will use the following estimate by Lewko \cite{lewko2026szemer}:


\begin{theorem}[\cite{lewko2026szemer}]\label{thm:line-point-incidence-strong}
    Let $\cP$ be a set of $m$ points in $\F_q^2$ and $\cL$ be a set of $n$ lines in $\F_q^2$. Let $p= \Char(\F_q)$ and if $mn \leq p^3$, then 
    \begin{equation*}
        |\cI(\cP,\cL)| \leq 3(mn)^{2/3} + m + n.
    \end{equation*}
\end{theorem}

\section{Agreement sets and conjugates} \label{sec:conjugates}

\subsection{Distinctness of Conjugates}

In this section, we define the conjugates of a bivariate polynomial $Q(X,Y) \in \F_p[X,Y]$ and prove some properties of the conjugates. At the high level,
the conjugates $Q_{i,b}(X,Y)$ of $Q(X,Y)$ are bivariate polynomials such that if 
$Q(X,Y)$ vanishes at many points of the form $(g^\alpha, \alpha)$, then so do (many of) the conjugates. They are based on the algebraic properties of the discrete log function.


\begin{definition}[Conjugates] \label{def:conjugates}
    Let $Q(X,Y) \in \F_p[X,Y]$ be a non-zero polynomial, then for any $i \in \{0,...,p-2\}$ and $b \in \{0,1\}$, define 
    \begin{equation*}
        Q_{i,b}(X,Y) = Q(g^iX,Y+i+b).
    \end{equation*}
    These are called the conjugates of $Q(X,Y)$. 
\end{definition}

These conjugates will serve as our ``list elements''. Our proof will proceed by showing that the agreement set cannot be too large if all of these are to be in the list. 

We formally define this agreement set below, along with some conjugate agreement sets.

\begin{definition}[Agreement set $T$] \label{def:agreement set}
    Let $Q(X,Y) \in \F_p[X,Y]$ be a non-zero polynomial, then define $T(Q) \subset \F_p$ as
    \begin{equation*}
        T(Q) = \{ \alpha \in \{0,...,N-1\} \text{ $|$ }Q(g^\alpha,\alpha) = 0\}.
    \end{equation*}
     Further for any $i \in [N-1]$, define
     \begin{align*}
        T(Q)_{i,0} &= \{ \alpha \in \{0,...,N-i-1\} \text{ $|$ }\alpha+i \in T(Q) \}\\
        T(Q)_{i,1} &= \{ \alpha \in \{N-i,..,N-1\} \text{ $|$ } \alpha +i -N \in T(Q) \}.
     \end{align*}
     In particular note that $|T(Q)| = |T(Q)_{i,0}| + |T(Q)_{i,1}|$.
\end{definition}

The latter set $T(Q)_{i, 1}$ is meant to handle the case when $\alpha +i$ exceeds $N$. Notice that for $\alpha \in T(Q)_{i,1}$, we have $\alpha + i - N = \alpha + i - (p-1) = \alpha + i + 1$.

Now we will show that $T(Q)_{i,b}$ contain the projection of the roots of $Q_{i,b}$ onto the second coordinate.

\begin{proposition}\label{prop:conjugates-and-S(Q)}
    For any $(i,b) \in [N-1] \times \{0,1\}$, let $\alpha \in T(Q)_{i,b}$. Then $Q_{i,b} (g^\alpha,\alpha) = 0$.
\end{proposition}

\begin{proof}
    Let $\alpha \in T(Q)_{i,b}$, then by definition $\alpha + i - b\cdot N \in T(Q)$. Evaluating $Q_{i, b}$ at $(g^{\alpha},\alpha)$, we get 
    \begin{align*}
        Q_{i,b}(g^\alpha,\alpha) &= Q(g^i\cdot g^{\alpha}, \alpha+i +b)\\
        &= Q(g^{\alpha +i - b\cdot N}, \alpha+i - b\cdot N) \qquad &\text{(Since $g^N =1$ and $N = -1 \text{ mod } p$})\\
        &=0 \qquad &\text{(Since $\alpha +i - b\cdot N \in T(Q)$)}
    \end{align*}
    which finishes the proof.
\end{proof}


\begin{lemma}[Distinctness of conjugates] \label{lem:conjugate-distinctness}
    Let $Q(X,Y) \in \F_p[X,Y] $ be a non-zero polynomial of degree less than $p$ such that $Q$ depends on both $X$ and $Y$. For any $(i,b),(i',b) \in \{0,..,N-1\}\times \{0,1\}$ where $i \neq i'$, we have $Q_{i,b}(X,Y) \neq \lambda Q_{i',b}(X,Y)$ for any $\lambda \in \F_p.$
\end{lemma}

\begin{proof}
    Let $Q(X,Y)$ be of the form
    \begin{equation*}
        Q(X,Y) = \sum_{i=0}^{p-1} Q^{(i)}(Y) X^i
    \end{equation*}
    where $Q^{(i)} \in \F_p[Y].$ Since $Q$ depends on both $X$ and $Y$, there exists some $k <p$ such that $Q^{(k)}(Y) \notin \F_p$. For the sake of contradiction suppose there exists $\lambda \in \F_p$ such that $Q_{(i,b)}(X,Y) = \lambda \cdot Q_{(i',b)}(X,Y)$. Then it follows that
    \begin{equation*}
        g^{ik} \cdot Q^{(k)}(Y+i+b)\cdot X^k = \lambda \cdot g^{i'k} \cdot  Q^{(k)}(Y+i' +b)\cdot X^{k}
    \end{equation*}
    by comparing monomials in $X$ of degree $k$. By comparing coefficients of the largest degree monomial in the above expression we get that $\lambda = g^{(i-i')k}$ and $Q^{(k)}(Y+i+b) = Q^{(k)}(Y+i'+b)$. However since $i \neq i'$, it follows that $Q^{(k)}(Y) = Q^{(k)}(Y+i'-i)$. Since $\deg(Q^{(k)}(Y)) <p$, this is only possible if $Q^{(k)}(Y) \in \F_p$. This contradicts our choice of $k$ and thus we are done.
\end{proof}

\subsection{Structure of agreement set $T(Q)$}

Here we will state the main observation about the geometry of $T(Q)$ when $Q$ is an irreducible polynomial. In particular we show that $T(Q)$ cannot contain long arithmetic progressions.

\begin{lemma}\label{lem:no-large-AP}
    Let $Q(X,Y) \in \F_p[X,Y]$ be an irreducible non-zero polynomial of degree $D$. Then $T(Q)$ cannot contain an arithmetic progression (over $\mathbb{Z}$) of size $D^2+2.$
\end{lemma}

\begin{proof}
    For the sake of contradiction assume $T(Q)$ has an arithmetic progression of size $D^2+2$. Let $\alpha$ be the starting point of the arithmetic progression and $\beta \neq 0$ be the common difference. Then,
    \begin{equation*}
        \{\alpha, \alpha + \beta,....,\alpha + (D^2+1)\cdot \beta\} \subset T(Q).
    \end{equation*}
    Since $T(Q) \subset \{0,..,N-1\}$ and the arithmetic progression is over the integers,
    \begin{equation*}
        \{\alpha , \alpha  + \beta ,...,\alpha + D^2 \cdot \beta\} \subset \{0,..,N-1-\beta\}.
    \end{equation*}
    By definition of $T(Q)_{\beta,0}$, it follows that 
    \begin{equation*}
        \{\alpha , \alpha  + \beta ,...,\alpha + D^2 \cdot \beta\} \subset T(Q)_{\beta,0}
    \end{equation*}
    and thus by \cref{prop:conjugates-and-S(Q)}, $\alpha + i\cdot \beta$ is a root of $Q_{\beta.0}(X,Y)$ for $0 \leq i \leq D^2$. Hence $Q(X,Y)$ and $Q_{\beta,0}(X,Y)$ have $D^2 +1$ common roots.

Recall that $Q(X,Y)$ is irreducible. By \cref{lem:conjugate-distinctness}, we know that $Q(X,Y) \neq \lambda \cdot  Q_{\beta,0}(X,Y)$ for any $\lambda \in \F_p$.
Thus $\gcd(Q,Q_{\beta,0}) =1$. By Bezout's theorem (\cref{thm:Bezout}), we get that $Q$ and $Q_{\beta,0}$ can have at most $D^2$ common roots, which is a contradiction.
\end{proof}

\begin{corollary}\label{cor:S(Q)-approximation-enhanced}
    Let $Q(X,Y) \in \F_p[X,Y]$ be an irreducible non-zero polynomial of degree at most $D$, then there exists a constant $c_D < N$ such that for any $c_D \leq i \leq N-1$, $|T(Q)_{i,0}| \geq |T(Q)| -   i\cdot(\log\log i)^{-2^{-2^{D^2 +11}}}. $
\end{corollary}

\begin{proof}
    Since $Q(X,Y)$ is an irreducible and non-zero polynomial of degree at most $D$, $Q_{i,1}(X,Y) = Q(g^i X, Y +i+1)$ is also an irreducible polynomial of degree at most $D$. Since $T(Q)_{i,1} \subset T(Q_{i,1})$ (\cref{prop:conjugates-and-S(Q)}), then by \cref{lem:no-large-AP}, $T(Q)_{i,1}$ does not contain any arithmetic progression of size $D^2+2$. Since $T(Q)_{i,1} \subset \{N-i,...,.N-1\}$, by \cref{thm:Szemeredi-Theorem-Gowers}, 
    \begin{equation*}
        |T(Q)_{i,1}| < i \cdot (\log \log i)^{-2^{-2^{D^2+11}}}
    \end{equation*}
    Further note that $|T(Q)| = |T(Q)_{i,0}| + |T(Q)_{i,1}|$ which implies
    \begin{equation*}
        |T(Q)_{i,0}| \geq |T(Q)| -   i\cdot(\log\log i)^{-2^{-2^{D^2 +11}}}.
    \end{equation*}
    which is the required relation. 
\end{proof}

\section{Approximation by univariate polynomials} \label{sec:univariate}

In this section, we consider low-degree univariate polynomials $f(X) \in \F_p[X]$. We are interested in the number of points $\alpha$ for which we can have $f(g^\alpha) = \alpha$, that is, $f$ agrees with $\dlog$ on $g^\alpha$.

\begin{theorem}\label{thm:univariate-approximation}
    For any $\epsilon >0$, there exists $n_{\epsilon}$ such that for any $N \geq n_{\epsilon}$, the following is true. Let $f\in \F_p[X]$ be a non-zero polynomial of degree $D$. Define $Q(X,Y)\in \F_p[X,Y]$ as $Q(X,Y) := Y - f(X)$. Then,
    \begin{equation*}
        |T(Q)|  \leq \sqrt{DN }\left( 1 + \epsilon + \cO(\log \log N)^{-C_{\epsilon}/2} \right)
    \end{equation*}
    where $C_{\epsilon} = 2^{-2^{\left(\frac{81}{4\epsilon^4}+11\right)}}$.
\end{theorem}

\begin{proof}
    Let $A = |T(Q)|$ for ease of notation. Let $m$ be a parameter which we will set later and set $d_{\epsilon} = \left\lceil\frac{9}{2\epsilon^2}\right\rceil$. Define a list $\cL_{m} \subset \F_p[X,Y]$ of polynomials as $\cL_{m} = \{ \hat{Q}_{i,0}(X,Y) \ | \ 1 \leq i \leq m \}$. For any $1 \leq i \leq m$, define $T_i = T(Q_{i,0}).$ 
    
    Our goal is to invoke the Johnson bound (\cref{lem:Johnson-Bound}) on the sets $\{T_i \}$. We have to show the hypotheses for it are satisfied. The proof will proceed based on two cases.

    \begin{itemize}
        \item \textbf{Case $D \geq d_{\epsilon}$ : } Note that $|T_i| \geq |T(Q)_{i,0}| = |T(Q)| - |T({Q})_{i,1}|.$ Since $T(Q)_{i,1} \subset \{N-i,...,N-1\} $, this implies $|T(Q)_{i,1}|\leq i$ which in turn implies
        \begin{equation*}
            |T_i| \geq A - i.
        \end{equation*}
        Hence for every $1\leq i \leq m$, we have 
        \begin{equation*}
            |T_i| \geq A - m.
        \end{equation*}
        Next we will show that for all $i \neq j$, $|T_i \cap T_j| \leq D$. For $1 \leq i \leq m$, define $f_i(X) = f(g^i\cdot X) - i$. Then for any $1 \leq i \leq m$,
        \begin{equation*}
            Q_{i,0}(X,Y) = Y - f_i(X).
        \end{equation*}
        By \cref{lem:conjugate-distinctness}, for any $i \neq j$, $f_i(X) \neq f_j(X)$. If $\alpha \in T_i \cap T_j$, then $f_i(g^\alpha) = f_j(g^\alpha) = \alpha$ and hence $g^\alpha$ is a root of $f_i(X) - f_j(X)$. Since $f_i(X) - f_j(X)$ is a non-zero polynomial of degree at most $D$, $f_i(X) - f_j(X)$ can have at most $D$ roots. Since for any $\alpha \ne \beta \in \{0,...,N-1\}, g^\alpha \neq g^\beta$, we get $|T_i \cap T_j| \leq D$.

        Now we can apply the Johnson bound (\cref{lem:Johnson-Bound}) on $T_1,...,T_m$, and get that
        \begin{equation*}
            A-m \leq \frac{N}{2m} + \sqrt{\frac{N^2}{4m^2} + ND}.
        \end{equation*}
        Setting $m = \sqrt{N/2}$, we get
        \begin{align*}
            A &\leq \sqrt{2N} +\sqrt{\frac{N}{2}+DN}\\
            &= \sqrt{DN}\cdot \left( \sqrt{1 + \frac{1}{2D}}+\sqrt{\frac{2}{D}} \right)\\
            &\leq \sqrt{DN}\cdot \left( 1 + \sqrt{\frac{1}{2D}} + \sqrt{\frac{2}{D}} \right) \qquad \left(\text{Since $\left(x^2+y^2\right) \leq \left(|x|+|y|\right)^2$}\right)\\
            &= \sqrt{DN} \cdot \left( 1 + \frac{3}{\sqrt{2D}} \right)\\
            &\leq \sqrt{DN}\cdot (1 + \epsilon) \qquad (\text{Since $D \geq d_{\epsilon}$}).
        \end{align*}
        Hence we have shown that 
        \begin{equation*}
            |T(Q)| \leq \sqrt{DN}\cdot \left( 1 + \epsilon +\cO(\log \log N)^{-C_{\epsilon}/2} \right)
        \end{equation*}
        which finishes this case.
        
        \item \textbf{Case $D < d_{\epsilon}$ : } We will use $D \leq d_{\epsilon},$ to get a better bound on $|T_i|$ than the previous case. Just like the previous case we know that $|T_i| \geq |T(Q_{i,0})|$. By \cref{cor:S(Q)-approximation-enhanced}, it follows that
        \begin{equation*}
            |T_i| \geq A - i\cdot(\log\log i)^{-2^{-2^{\left(\frac{81}{4\epsilon^4}+11\right)}}}.
        \end{equation*}
        Thus for all $1\leq i \leq m$, it follows that 
        \begin{equation*}
            |T_i| \geq A - m\cdot(\log \log m)^{-C_{\epsilon}} 
        \end{equation*}
        where $C_\epsilon = 2^{-2^{\left(\frac{81}{4\epsilon^4}+11\right)}}.$

        For any $i\neq j$, just like the previous case, it follows that $|T_i \cap T_j|\leq D$. Now we apply the Johnson Bound (\cref{lem:Johnson-Bound}) on $T_1,...,T_m$ and get that
        \begin{equation*}
            A - m\cdot (\log \log m)^{-C_{\epsilon}} \leq \frac{N}{2m} +\sqrt{\frac{N^2}{4m^2} + DN}.
        \end{equation*}
        Rearranging the terms we get 
        \begin{align*}
            A &\leq \sqrt{DN}\cdot \left(  \sqrt{1 + \frac{N}{4m^2D}}+\frac{\sqrt{N}}{2m\sqrt{D}} + \frac{m\cdot (\log \log m)^{-C_{\epsilon}}}{\sqrt{DN}}\right)\\
            &\leq \sqrt{DN}\cdot \left( 1 + \frac{\sqrt{N}}{m\sqrt{D}} + \frac{m\cdot (\log\log m)^{-C_{\epsilon}}}{\sqrt{DN}}\right) \qquad \left(\text{Since $\left(x^2+y^2\right) \leq \left(|x|+|y|\right)^2$ }\right).
        \end{align*}
        Setting $m = \sqrt{N}\cdot \left(\log \log N\right)^{C_{\epsilon}/2}$, we get 
        \begin{align*}
            A &\leq \sqrt{DN}\cdot \left( 1 + \frac{1}{\sqrt{D}\cdot (\log\log N)^{C_{\epsilon}/2}} + \frac{(\log\log N)^{C_{\epsilon}/2}\cdot \left(  \log \log \left(\sqrt{N}\cdot (\log \log N )^{C_{\epsilon}/2}\right) \right)^{-C_{\epsilon}}}{\sqrt{D}}\right)\\
            &\leq \sqrt{DN}\cdot \left( 1 + \cO(\log\log N)^{-C_{\epsilon}/2} \right)  \qquad \text{(Since $d \geq 1$ and $N$ is large enough)}
        \end{align*}
        Hence we have shown that
        \begin{equation*}
            |T(Q)| \leq \sqrt{DN}\cdot \left(1 + \epsilon + \cO(\log \log N)^{-C_{\epsilon}/2} \right)
        \end{equation*}
        which finishes this case.
    \end{itemize}
\end{proof}

\begin{remark}
    The lower-order term $\sqrt{DN}\cdot (\log\log N)^{-C_{\epsilon}/2}$ may be improved using stronger quantitative estimates for Szem\'eredi's Theorem, for instance \cite{GT17} and \cite{leng2024improved}. For ease of presentation we do not invoke those in this paper.
\end{remark}


\section{Approximation by bivariate polynomials} \label{sec:bivariate}

In this section, we will be interested in the number of locations where a bivariate polynomial $Q(X, Y)$ may satisfy $Q(g^\alpha, \alpha) = 0$. Recall from \cref{def:agreement set} that this set is denoted $T(Q)$.

\begin{theorem}\label{thm:bivariate-approximation}
    For any $\epsilon >0$, there exists $n_{\epsilon}$ such that for any $N \geq n_{\epsilon}$, if $Q(X,Y) \in \F_p[X,Y]$ is a non zero polynomial of degree $D$ then,
    \begin{equation*}
        |T(Q)| \leq D\sqrt{N} \left( 1  +\epsilon + \cO( \log\log N)^{-C_{\epsilon}/2} \right).
    \end{equation*}
    where $C_\epsilon = 2^ { -2^{\left(\frac{9}{2\epsilon^2}+11\right)} }.$
\end{theorem}

\begin{proof}
    Let $Q(X,Y) = Q^{(1)}(X,Y) \dotsb Q^{(\ell)}(X,Y)$ be the irreducible factorization of $Q(X,Y)$ over $\F_p[X,Y]$. (The polynomials $\{Q^{(i)}\}$ need not be distinct.) By an averaging argument, there exists $i_0 \in [\ell]$ such that 
    \begin{equation} \label{eqn:Dense_Factor}
        |T(Q^{(i_0)})|\geq |T(Q)|\frac{\deg (Q^{(i_0)})}{\deg(Q)}.
    \end{equation}
    
    For ease of notation let $\hat Q = Q^{(i_0)}$, $A = |T(\hat Q)|$ and $d = \deg(\hat Q)$.
    
    Let $m$ be a parameter which we will set later and set $d_{\epsilon} = \left\lceil \frac{3}{\sqrt{2}\epsilon} \right\rceil$. Define $\cL_{m} \subset \F_q[X,Y]$ where $\cL_{m} = \{\hat Q_{i,0}(X,Y) \mid 1 \leq i \leq m\}$. For any $1 \leq i \leq m$, define $T_i = T(Q_{i,0})$. The proof will move ahead based on two cases, similar to the univariate case. Once again, our goal is to prove the hypotheses for the Johnson bound are satisfied, and then invoke it.
    \begin{itemize}
        \item \textbf{Case $d \geq d_{\epsilon}$ : } Note that $|T_i| \geq |T(\hat{Q})_{i,0}| = |T(\hat Q)| - |T(\hat{Q})_{i,1}|.$ Since $T(\hat Q)_{i,1} \subset \{N-i,...,N-1\} $, this implies $|T(\hat Q)_{i,1}|\leq i$ which in turn implies
        \begin{equation*}
            |T_i| \geq A - i
        \end{equation*}
        Hence for all $1\leq i \leq m$, we have 
        \begin{equation*}
            |T_i| \geq A - m.
        \end{equation*}
        Next we will show that for all $i \neq j$, $|T_i \cap T_j| \leq d^2$. We know that $\hat Q$ is an irreducible polynomial of degree $d$ which implies that for any $1 \leq i \leq m$, $\hat Q_{i,0}$ is an irreducible polynomial of degree $d$. For any $i\neq j$, \cref{lem:conjugate-distinctness} shows that $\hat Q_{i,0} \neq \lambda\cdot \hat Q_{j,0}$ for any $\lambda \in \F_p$. This is equivalent to showing that $\gcd(\hat Q_{i,0}, \hat Q_{j,0})=1$. By definition of $T_i's$ it is clear that $T_i \cap T_j$ is the set of common zeros of $\hat Q_{i,0}$ and $\hat Q_{j,0}$. Thus by \cref{thm:Bezout}, it follows that $|T_i \cap T_j| \leq d^2.$

        Now we apply the Johnson Bound (\cref{lem:Johnson-Bound}) on $T_1,...,T_m$ and get that 
        \begin{equation*}
            A-m \leq \frac{N}{2m} + \sqrt{\frac{N^2}{4m^2} + Nd^2}.
        \end{equation*}
        Now we set $m = \sqrt{n/2}$ and get
        \begin{align*}
            A &\leq \sqrt{2N} + \sqrt{\frac{N}{2}+ Nd^2}\\
            &= d\sqrt{N}\cdot\left( \sqrt{1 + \frac{1}{2d^2}}+\sqrt{\frac{2}{d^2}} \right)\\
            &\leq d\sqrt{N} \cdot \left(1 +  \frac{1}{d \sqrt{2}} + \frac{\sqrt{2}}{d} \right) \qquad \left(\text{Since $\left(x^2+y^2\right) \leq \left(|x|+|y|\right)^2$}\right)\\
            &= d\sqrt{N}\cdot \left( 1 + \frac{3}{\sqrt{2}d} \right)\\
            &\leq d \sqrt{N} \cdot (1 + \epsilon) \qquad (\text{Since $d \geq d_{\epsilon}$}).
        \end{align*}
        \cref{eqn:Dense_Factor} would then imply
        \begin{equation*}
            |T(Q)| \leq D\sqrt{N}\cdot (1 + \epsilon).
        \end{equation*}
        Thus we are done with this case.

        \item \textbf{Case $d < d_{\epsilon}$ : } We will use $d \leq d_\epsilon$ to get a better bound on $|T_i|$ than the previous case. Just like the previous case we know that $|T_i| \geq |T(\hat Q)_{i,0}|$. By \cref{cor:S(Q)-approximation-enhanced}, it follows that
        \begin{equation*}
            |T_i| \geq A - i\cdot (\log \log i)^{-2^ {- 2^{\left(\frac{9}{2\epsilon^2}+11\right)} } }
        \end{equation*}
        Thus for all $1 \leq i \leq m$, it follows that
        \begin{equation*}
            |T_i| \geq A - m\cdot (\log \log m)^{-C_\epsilon}  
        \end{equation*}
        where $C_\epsilon = 2^ { -2^{\left(\frac{9}{2\epsilon^2}+11\right)} }.$
        
        For any $i \neq j$, $|T_i \cap T_j|\leq d^2$. The argument is identical to the one in the previous case. Now we apply the Johnson Bound (\cref{lem:Johnson-Bound}) on $T_1,...,T_m$ and get that
        \begin{equation*}
            A - m\cdot (\log \log m)^{-C_\epsilon } \leq \frac{N}{2m} + \sqrt{\frac{N^2}{4m^2} + N d^2}.
        \end{equation*}
        Rearranging the terms we get
        \begin{align*}
            A &\leq d\sqrt{N}\cdot \left( \sqrt{1+\frac{N}{4m^2d^2}}+\frac{\sqrt{N}}{2md} + \frac{m \cdot (\log \log m)^{-C_\epsilon}}{d \sqrt{N}}\right)\\
            &\leq d \sqrt{N}\cdot \left( 1 + \frac{\sqrt N}{md} + \frac{m\cdot (\log\log m)^{-C_{\epsilon}}}{d \sqrt{N}} \right) \qquad \left(\text{Since $\left(x^2+y^2\right) \leq \left(|x|+|y|\right)^2$}\right)
        \end{align*}
        Setting $m = \sqrt{N}\cdot \log\log(N)^{C_{\epsilon}/2}$ we get
        \begin{align*}
            A &\leq d\sqrt{N} \cdot \left( 1 + \frac{1}{d \cdot (\log \log N)^{C_{\epsilon}/2}} + \frac{(\log\log N)^{C_{\epsilon}/2} \cdot  \left( \log\log(N^{1/2}\cdot (\log\log N))^{C_{\epsilon}/2} \right)^{-C_{\epsilon}}}{d} \right)\\
            &\leq d \sqrt{N}\cdot \left( 1+ \cO (\log\log N)^{-C_{\epsilon}/2}\right) \qquad \text{(Since $d \geq 1$ and $N$ is large enough)}
        \end{align*}
        \cref{eqn:Dense_Factor} would then imply
        \begin{equation*}
            |T(Q)|\leq D\sqrt{N}\cdot \left( 1 + \cO (\log \log N)^{- C_{\epsilon}/2} \right).
        \end{equation*}
        Hence, we are done with this case as well.
    \end{itemize}
\end{proof}

\begin{remark}
    Similar to the univariate case, stronger estimates for Szemeredi's theorem improve the lower-order term $D\sqrt{N}\cdot (\log\log N)^{-C_{\epsilon}/2}$. 
\end{remark}

\section{Approximation by linear functions} \label{sec:linear}


In this section we prove bounds on approximating the discrete log by polynomials of degree $1$, namely linear functions.

\begin{theorem}\label{thm:linear-main}
    Let $f(X) \in \F_p[X]$ have degree $\leq 1$.
    Then:
    \[ |\{ x \in \F_p^* \mid  f(x) = \dlog(x)| \leq O(p^{2/5}).
    \]
 \end{theorem}

The strategy, of proving properties of the discrete log function via the line-point incidence theorem over finite fields, comes from the analysis of Bourgain's two-source extractor~\cite{Bourgain} (made especially clear in Rao's exposition~\cite{Rao}).

\subsection{Proof Overview}
We first give an overview of the argument.

Let $S \subseteq \F_p^2$ denote the intersection of the graphs of $f(X)$ and $\dlog$, that is:
$$S= \{ (x, y) \in \F_p^2 \ : x \neq 0, y = \dlog(x), y = f(x)\}.$$ 
Thus $S$ is contained in a line in $\F_p^2$. Let $T\subseteq \F_p$ be the projection of $S$ to the first coordinate:
$$ T = \{ x \in \F_p^* \mid \dlog(x) = f(x) \}.$$
Now consider the iterated sumset $S+ S+ S$.

On the one hand, $S+S+S$ is contained in a line in $\F_p^2$: indeed, if $f(X) = aX+ b$, 
then every element $(x,y)$ of $S+S+S$ satisfies $y = ax+ 3b$. This means that $|S+S+S| \leq p$.

On the other hand 
$$S+S+S = \{  (x_1 + x_2 + x_3,   \dlog(x_1) + \dlog(x_2) + \dlog(x_3)) \mid x_1, x_2 , x_3 \in T \}.$$
Let us pretend for the moment that $\dlog$ satisfies the exact homomorphism identity $\dlog(u) + \dlog(v) = \dlog(uv)$.
Then $S+S+S = \{ x_1 + x_2 + x_3, \dlog(x_1 x_2 x_3) \mid x_1,x_2,x_3 \in T \}$, 
and thus $S+S+S$ has the same size as the set $\cP$ given by:
$$ \cP = \{ (x_1 + x_2 + x_3, x_1 x_2 x_3) \mid x_1, x_2, x_3 \in T \}.$$
We now define a set $\cL$ of lines in $\F_p^2$. 
For each $\{x_1, x_2 \} \subseteq T$, define the line $L_{x_1, x_2} = \{ x_1 + x_2 + z, x_1x_2z) \mid z \in \F_p\}$.
There are ${|T| \choose 2 }$ such lines, and each is incident to $|T|$ points of $\cP$.
Putting this incidence data about $\cL$ and $\cP$ into the line-point incidence theorem, along with our upper bound $|\cP| \leq p$, gives the desired upper bound on $|T|$.

This concludes the proof overview. 
The only gap in this is the wished-for exact homomorphic property of $\dlog$, which is not quite true because of wrap-around. To get around this, we will zoom in on subsets of $S$ and $T$ where some version of the homomorphic property holds.

\subsection{Proof of \Cref{thm:linear-main}}

We start by partitioning $S$ (and $T$) into $m$ buckets, for a parameter $m$ which we will fix later (to equal $2$ or $3$). 

\noindent Define $S_{i, m}$ by:
\[ S_{i,m} := \left\{ (x, \dlog(x)) \in S \ : \ \frac{i-1}{m}(p-1) \leq \dlog(x) < \frac{i}{m}(p-1) \right\} \]
then note that $S_{i,m}$ for $i \in [m]$ partition $S$.

Further, let $T$ and $T_{i,m}$ denote the respective projections of these sets to their first coordinate. That is, $T = \{x : (x, \dlog(x)) \in S\}$ and similarly for $T_{i, m}$.

We define a $k$-fold ``sum and product'' operation $U^{\star k}$ on subsets $U \subset \F_p$. This operation is motivated by the property of the (standard) logarithm: $\log(ab) = \log(a) + \log (b)$.

\begin{definition}
    For sets $U \subseteq \F_p$ and $k \in \N$, let 
    \[ U^{\star k} = \{ (u_1 + u_2 + \dotsb + u_k, u_1 u_2 ... u_k) \in \F_p^2 : u_i \in U \ \forall \ i\in [k] \}. \]
\end{definition}
We will now prove some properties of the sets $T_{i, m}$ under this operation. 

\begin{lemma} \label{lem:star size}
    For any $i, m$, $|T_{i, m}^{\star m}| = |m \cdot S_{i, m}|$.
\end{lemma}
Recall that $m \cdot V$ denotes the $m$-fold sumset $V+V+ \ldots + V$.
\begin{proof}
We prove this by defining a bijection 
\[
\begin{split}
    \phi : T_{i, m}^{\star m} & \to m \cdot S_{i, m} \\
    (a, b) & \mapsto (a, \dlog(b) - (i-1)).
\end{split}
\]

We first have to show this map is well-defined, that is, the image of any element under the map lies in $m \cdot S_{i, m}$ as claimed. Take $(a, b) \in T_{i, m}^{\star m}$. Then, there exist $a_1, \dots, a_m \in T_{i, m}$ such that $a = a_1 + \dotsb + a_m$, $b = a_1 a_2 \dots a_m$. Let $b_j = \dlog(a_j)$ for each $j$; we are left to show $\dlog(b) - (i-1) = b_1 + \dotsb + b_m$. Since each $a_j \in T_{i, m}$, we have $\frac{i-1}{m}(p-1) \leq b_j < \frac{i}{m}(p-1)$. Then, $(i-1)(p-1) \leq \sum b_j < i(p-1) $. 
\[ 
\begin{split}
    \dlog(b) & = \dlog(a_1 a_2 \dots a_m) \\
    & = \dlog(a_1) + \dotsb \dlog(a_m) - (i-1)(p-1) \text{ since } (i-1)(p-1) \leq \sum \dlog(a_j) < i(p-1)\\
    & \equiv b_1 + \dotsb + b_m + (i-1) \mod p.
\end{split}
\]

This shows the map is well-defined. Now we are left to show injectivity and surjectivity.

For injectivity, suppose $\phi(a_1, b_1) = (a_2, b_2)$. Since $\phi$ fixes the first coordinate, we have $a_1 = a_2$. Further, we have $\dlog(b_1) = \dlog(b_2)$ which implies $b_1 = b_2$.

For surjectivity, let $(a', b') = (a_1 + \dotsb + a_m, \dlog(a_1) + \dotsb + \dlog(a_m)) \in m \cdot S_{i, m}$. By a similar argument to the well-definedness case, we can pick $(a, b)$ with $a = a'=  \sum a_j$ and $b= \prod_j a_j$, which would give $\phi(a, b) = (a', b')$ as required. 
\end{proof}

\begin{proposition} \label{prop:size of 2S}
    Let $S \subseteq \F_p^2$ be the agreement set as above and $T$ its projection onto the first coordinate. Then, $|2 \cdot S| \geq \frac{|T|^2}{8}$.
\end{proposition}

\begin{proof}
    Recall that $T$ may be subdivided into $T_{1, 2}$ and $T_{2, 2}$. Without loss of generality we assume $|T_{1, 2}| \geq |T|/2$, otherwise in the below argument replace $T_{1, 2}$ with $T_{2,2}$.

Now since $S_{1, 2} \subseteq S$, 
\[
|2 \cdot S| \geq |2 \cdot S_{1, 2}| = |T_{1, 2}^{\star 2}|
\]
    where the equality follows from \cref{lem:star size}.

We also have that $|T_{1, 2}^{\star 2}|$ is at least the number of pairs (with repetition) up to ordering in $T_{1, 2}$. This is because an element of $|T_{1, 2}^{\star 2}|$ is a pair of the form $(a_1 + a_2, a_1 a_2)$, which determines $a_1$ and $a_2$ in $T_{1, 2}$ up to labeling. Hence, 
\[ |T_{1, 2}^{\star 2}| \geq \frac{|T_{1, 2}|^2}{2} \]

Combining with the above, we have 
\[ |2 \cdot S| \geq \frac{|T_{1, 2}|^2}{2} \geq \frac{|T|^2}{8} \]

    which completes the proof.
\end{proof}

We are now ready to prove the main theorem of the section.

\begin{proof}[Proof of \cref{thm:linear-main}]
    Let $c$ be a constant that we will fix later; we will eventually show $|S| < c (p^{1/2-\eps})$ for some $\eps>0$.


    Assume for contradiction that $|S| \geq c (p^{1/2 -\eps})$. Divide $S$ into $S_{1, 3}, S_{2, 3}$ and $S_{3,3}$ as defined above (that is, $m=3$). At least one of these has size at least $|S|/3$; replace $S$ by that set $S_{i, 3}$ and drop any extraneous elements so that $|S| = \frac{c}{3} (p^{1/2 - \eps})$ exactly.

    Let $T$ be the projection of $S$ to the first coordinate, and thus $|T| = |S|$. Recall that $T^{\star 2}$ consists of pairs; take a pair $(a, b) \in T^{\star 2} $ and define the line $L_{a, b} = \{ (a+ z, bz) \ : \ z \in \F\}$. Let $\calL = \{L_{a, b} : (a, b) \in T^{\star 2}\}$. This will be the set of lines that we invoke the line-point incidence result with. Observe that $|\calL| = |T^{\star 2}|$.

    Now, since every $(a, b) \in T^{\star 2} $ comes from two elements in $T$ up to ordering, we have 
    \[
    \frac{|T|^2}{2} \leq |T^{\star 2}| \leq |T|^2
    \]

    Substituting the bounds on $T$ as per the contradiction assumption, and the observation about $\calL$,
    \[
    \frac{c^2 p^{1-2\eps}}{18} \leq |\calL| \leq \frac{c^2 p^{1-2\eps}}{9}.
    \]

    The set of points $\calP$ we will use will be $T^{\star 3}$. By \cref{lem:star size}, $|T^{\star 3}| = |3 \cdot S|$. Recall that $S$ was a subset of the original line that purported to approximate $\dlog$. Hence, $3 \cdot S$ is also a subset of some line, and hence $|3 \cdot S| \leq p$. Then, we can conclude that $|T ^{\star 3}| \leq p$.

    We need to lower-bound the size of $\calP$ as well. We have

    \[
    |T^{\star 3}| = |3 \cdot S| \geq |2 \cdot S| \geq \frac{|T|^2}{8}.
    \]

    where the last inequality is by \cref{prop:size of 2S}. Then, combined with the contradiction assumption, this gives 
    
    \[
    \frac{c^2}{72}p^{1-2\eps} \leq |\calP| \leq p.
    \]

    Now, $\calP$ consists of points $( a_1 + a_2 + a_3, a_1 a_2 a_3)$ while $\calL$ consists of lines of the form $\{( a_1 + a_2 + z, a_1 a_2 z) : z \in \F_p\}$. From this form, it is clear that any $z \in T$ and $L_{a, b} \in \calL$ gives an incidence. Then, 
    \[
    \calI(\calP, \calL) \geq |T| |\calL| \geq \frac{c^3 p^{3/2-3\eps}}{54}
    \]


    For a large enough prime $p$, we have $|\cP|\cdot |\cL| \leq \frac{c^2}{9}p^{2-2\epsilon} \leq p^3$, which is the hypothesis of \cref{thm:line-point-incidence-strong}.

    Then, we can invoke the bound to conclude 
    \begin{align*}
        |\cI(\cP,\cL)| &\leq 3(|\cL|\cdot |\cP|)^{2/3} + |\cL| + |\cP|\\
        &\leq 3\cdot \left(\frac{c^2}{9}p^{2-2\epsilon}\right)^{2/3} + \left( \frac{c^2 p^{1-2\epsilon}}{9} \right) + p\\
        &\leq 3\left(\frac{c^2}{9}\right)^{2/3} p^{\frac{4(1-\epsilon)}{3}} + O(p)
    \end{align*}
which is a contradiction for $\eps = 1/10$ and a large enough constant $c$.
    
\end{proof}
 
\bibliographystyle{alpha}
\bibliography{ref.bib}

@article{CS00,
	author = {Coppersmith, Don and Shparlinski, Igor},
	date = {2000/06/01},
	doi = {10.1007/s001450010002},
	id = {Coppersmith2000},
	isbn = {1432-1378},
	journal = {Journal of Cryptology},
	number = {3},
	pages = {339--360},
	title = {On Polynomial Approximation of the Discrete Logarithm and the {D}iffie--{H}ellman Mapping},
	url = {https://doi.org/10.1007/s001450010002},
	volume = {13},
	year = {2000}}

@article{Winterhof,
  author  = {Winterhof, Arne},
  title   = {Polynomial Interpolation of the Discrete Logarithm},
  journal = {Designs, Codes and Cryptography},
  volume  = {25},
  number  = {1},
  pages   = {63--72},
  year    = {2002},
  doi     = {10.1023/A:1012556500517}
}

@article{KW,
  author  = {Kiltz, Eike and Winterhof, Arne},
  title   = {Polynomial interpolation of cryptographic functions
             related to {Diffie--Hellman} and discrete logarithm problem},
  journal = {Discrete Applied Mathematics},
  volume  = {154},
  number  = {2},
  pages   = {326--336},
  year    = {2006},
  doi     = {10.1016/j.dam.2005.03.030}
}

@article{IW,
  author  = {I{\c{s}}{\i}k, Leyla and Winterhof, Arne},
  title   = {{Carlitz} rank and index of permutation polynomials},
  journal = {Finite Fields and Their Applications},
  volume  = {49},
  pages   = {156--165},
  year    = {2018},
  doi     = {10.1016/j.ffa.2017.09.002},
  eprint  = {1611.06361},
  archivePrefix = {arXiv}
}

@misc{BW,
  author  = {Bienvenu, Pierre-Yves and Winterhof, Arne},
  title   = {Additive index and {Carlitz} rank},
  year    = {2026},
  eprint  = {2604.25414},
  archivePrefix = {arXiv},
  url     = {https://arxiv.org/abs/2604.25414}
}

@article{NW,
  author  = {Niederreiter, Harald and Winterhof, Arne},
  title   = {Incomplete Character Sums and Polynomial Interpolation
             of the Discrete Logarithm},
  journal = {Finite Fields and Their Applications},
  volume  = {8},
  number  = {2},
  pages   = {184--192},
  year    = {2002},
  doi     = {10.1006/ffta.2001.0334}
}

@article{Gow01,
	author = {Gowers, W. T. },
	date = {2001/08/01},
	doi = {10.1007/s00039-001-0332-9},
	id = {Gowers2001},
	isbn = {1420-8970},
	journal = {Geometric \& Functional Analysis GAFA},
	number = {3},
	pages = {465--588},
	title = {A new proof of {S}zemer{\'e}di's theorem},
	url = {https://doi.org/10.1007/s00039-001-0332-9},
	volume = {11},
	year = {2001}}

@book{fulton1989algebraic,
  title     = {Algebraic Curves: An Introduction to Algebraic Geometry},
  author    = {Fulton, William},
  year      = {1989},
  publisher = {Addison-Wesley Publishing Company},
  series    = {Advanced Book Classics},
  isbn      = {978-0201510102}
}

@article{GT17,
author = {Green, Ben and Tao, Terence},
title = {NEW BOUNDS FOR {S}ZEMER\'eDI'S THEOREM, III: A POLYLOGARITHMIC BOUND FOR},
journal = {Mathematika},
volume = {63},
number = {3},
pages = {944-1040},
doi = {https://doi.org/10.1112/S0025579317000316},
url = {https://londmathsoc.onlinelibrary.wiley.com/doi/abs/10.1112/S0025579317000316},
eprint = {https://londmathsoc.onlinelibrary.wiley.com/doi/pdf/10.1112/S0025579317000316},
year = {2017}
}

@article{leng2024improved,
  title={Improved Bounds for {S}zemer\'edi's Theorem},
  author={Leng, James and Sah, Ashwin and Sawhney, Mehtaab},
  journal={arXiv preprint arXiv:2402.17995},
  year={2024}
}

@article{GoldreichRS00,
  author       = {Oded Goldreich and
                  Ronitt Rubinfeld and
                  Madhu Sudan},
  title        = {Learning Polynomials with Queries: The Highly Noisy Case},
  journal      = {{SIAM} J. Discret. Math.},
  volume       = {13},
  number       = {4},
  pages        = {535--570},
  year         = {2000},
  url          = {https://doi.org/10.1137/S0895480198344540},
  doi          = {10.1137/S0895480198344540},
  bibsource    = {dblp computer science bibliography, https://dblp.org}
}

@ARTICLE{BourgainKT2004,
  title     = "A sum-product estimate in finite fields, and applications",
  author    = "Bourgain, Jean and Katz, Nets and Tao, Terence",
  journal   = "Geom. Funct. Anal.",
  publisher = "Springer Science and Business Media LLC",
  volume    =  14,
  number    =  1,
  pages     = "27--57",
  month     =  feb,
  year      =  2004
}

@article{Bourgain,
author = {Bourgain, J.},
title = {MORE ON THE SUM-PRODUCT PHENOMENON IN PRIME FIELDS AND ITS APPLICATIONS},
journal = {International Journal of Number Theory},
volume = {01},
number = {01},
pages = {1-32},
year = {2005},
doi = {10.1142/S1793042105000108},

URL = { 
        https://doi.org/10.1142/S1793042105000108
},
eprint = { 
        https://doi.org/10.1142/S1793042105000108
}
}

@inproceedings{Rao,
  title={An exposition of {B}ourgain’s 2-source extractor},
  author={Rao, Anup},
  booktitle={Electronic Colloquium on Computational Complexity (ECCC)},
  volume={14},
  number={034},
  year={2007},
  url = {https://eccc.weizmann.ac.il/report/2007/034/}
}

@article{KedlayaKopparty,
  title={On the degree of polynomials computing square roots mod p},
  author={Kedlaya, Kiran and Kopparty, Swastik},
  journal={arXiv preprint arXiv:2311.10956},
  year={2023}
}

@ARTICLE{DiffieHellman,
  author={Diffie, W. and Hellman, M.},
  journal={IEEE Transactions on Information Theory}, 
  title={New directions in cryptography}, 
  year={1976},
  volume={22},
  number={6},
  pages={644-654},
  doi={10.1109/TIT.1976.1055638}}

@article{brakensiek2026algorithmic,
  title={Algorithmic List Decoding of {R}eed-{S}olomon Codes up to Capacity},
  author={Brakensiek, Joshua and Chen, Yeyuan and Putterman, Aaron and Zhang, Zihan and Zheng, Kai Zhe},
  journal={arXiv preprint arXiv:2609.08005},
  year={2026}
}

@article{lewko2026szemer,
  title={A {S}zemer\'edi-{T}rotter Theorem in Arbitrary Fields},
  author={Lewko, Mark},
  journal={arXiv preprint arXiv:2609.27023},
  year={2026}
}

@article{miao2026szemer,
  title={The {S}zemer\'edi-{T}rotter Estimate in Finite Field with its Applications},
  author={Miao, Changxing and Xie, Rui},
  journal={arXiv preprint arXiv:2609.35190},
  year={2026}
}

@article{yu2026two,
  title={Two proofs and a tight {S}zemer\'edi--{T}rotter theorem},
  author={Yu, Hung-Hsun Hans},
  journal={arXiv preprint arXiv:2610.10649},
  year={2026}
}

@article{harsha2026reed,
  title={{R}eed-{S}olomon Codes at Capacity: Algorithmic List-Decoding and Proximity Gaps},
  author={Harsha, Prahladh and Kumar, Mrinal and Saptharishi, Ramprasad},
  journal={arXiv preprint arXiv:2610.08610},
  year={2026}
}

\end{document}